\documentclass[12pt]{article}
\usepackage{euscript,amsmath,amssymb,latexsym,amsthm}

\DeclareMathOperator{\St}{St}

\theoremstyle{remark}
\newtheorem*{remark}{Remark}

\theoremstyle{plain}
\newtheorem{theorem}{Theorem}
\newtheorem{proposition}{Proposition}

\begin{document}
\title{
\vspace{1cm} {\bf Derivation of the particle dynamics from kinetic equations}
}
\author{A.\,S.~Trushechkin\bigskip
 \\
{\it  Steklov Mathematical Institute of the Russian Academy of Sciences}
\\ {\it Gubkina St. 8, 119991 Moscow, Russia}\medskip\\
{\it  National Research Nuclear University ``MEPhI''}
\\ {\it Kashirskoe Highway 31, 115409 Moscow, Russia}\bigskip
\\ e-mail:\:\texttt{trushechkin@mi.ras.ru}}

\date {}
\maketitle

\begin{abstract}
We consider the microscopic solutions of the Boltzmann--Enskog equation discovered by Bogolyubov. The fact that the time-irreversible kinetic equation has time-reversible microscopic solutions is rather surprising. We analyze this paradox and show that the reversibility or irreversibility property of the Boltzmann--Enskog equation depends on the considered class of solutions. If the considered solutions have the form of sums of delta-functions, then the equation is reversible. If the considered solutions belong to the class of continuously differentiable functions, then the equation is irreversible. Also, we construct the so called approximate microscopic solutions. These solutions are continuously differentiable and they are reversible on bounded time intervals.

This analysis suggests a way to reconcile the time-irreversible kinetic equations with the time-reversible particle dynamics. Usually one tries to derive the kinetic equations from the particle dynamics. On the contrary, we postulate the Boltzmann--Enskog equation or another kinetic equation and treat their microscopic solutions as the  particle dynamics. So, instead of the derivation of the kinetic equations from the microdynamics we suggest a kind of derivation of the microdynamics from the kinetic equations.
\end{abstract}

\section{Introduction}

In the known works of Bogolyubov \cite{Bogol75, BogBog}, the Boltzmann--Enskog kinetic equation is shown to have microscopic solutions. This result is surprising for two reasons. The first reason is that the Boltzmann--Enskog equation is regarded as an approximation to the particle dynamics in the case of a dilute gas. This result shows that the Boltzmann--Enskog equation has exact microscopic solutions as well. Also the Vlasov kinetic equation is known to have microscopic solutions \cite{Vlasov}.

But, in contrast to the Vlasov equation, the Boltzmann--Enskog equation is time-irreversible and describes the entropy production. However, the microscopic dynamics is time-reversible. So, this result tells that an irreversible equation has reversible solutions. The aim of this report is to analyze this paradox.

We show that the reversibility or irreversibility property of the Boltzmann--Enskog equation depends on the considered class of solutions. If the considered solutions have the form of sums of delta-functions, then the equation is reversible. If the considered solutions belong to the class of continuous functions, then the equation is irreversible. 

Also, we construct the so called approximate microscopic solutions of the Boltzmann--Enskog equation. These solutions are continuous and they are reversible on bounded time intervals. The approximate microscopic solutions turn to the (exact) microscopic solutions in the limit.

This result can be compared with the Kozlov's diffusion theorems for a collisionless continuous medium in a rectangular box \cite{Kozlov}. The dynamics is described by the Liouville equation. If we consider a solution in the form of a delta-function, then the solution is periodic or almost periodic (this case is reduced to the single-particle dynamics). But if we consider an integrable solution, the spatial density weakly converges to the uniform distribution, as time indefinitely increases. So, the asymptotic dynamical properties of solutions depend on the considered class of functions.

But anyway, the Liouville equation is time-reversible. Here we show that in the case of the Boltzmann--Enskog equation, not only the asymptotic properties of solutions, but also the reversibility or irreversibility property of them depend on the considered class of functions.

Further, this analysis suggests a way to reconcile irreversible kinetic equations with the reversible microscopic dynamics. Kinetic equations describe the gas as a whole, while equations of microdynamics describe the same gas as a collection of particles (atoms or molecules). If we believe that the properties of the gas as a whole can be reduced to the properties of its particles (the position of reductionism), then we shall postulate the equations of microdynamics and try to derive the corresponding kinetic equations from them.

An elegant derivation of the Boltzmann and Boltzmann--Enskog equations from the Liouville equation was proposed by Bogolyubov \cite{Bogol}. He uses the BBGKY (Bogolyubov--Born--Green--Kirkwood--Yvon) hierarchy of equations, the thermodynamic limit and some additional assumptions (which do not follow from the Liouville equation or its initial data). This derivation leads to the divergences in high order corrections to the Boltzmann and Boltzmann--Enskog equations \cite{Bogol77}.

Another derivation of the Boltzmann equation was proposed by Lanford \cite{Lanford} (see also \cite{Lebowitz,Spohn}). He also uses the BBGKY hierarchy of equations, the Boltzmann--Grad limit, all assumptions are included in the initial conditions for the Liouville equation. But the derivation can be applied only to small times. On the contrary, the Boltzmann equation is interesting from the viewpoint of the large time asymptotics (the relaxation to the Maxwell distribution) \cite{Villani}.

So, we have not a satisfactory derivation of the Boltzmann and Boltzmann--Enskog equations from the equations of microdynamics.

In this report we propose an alternative way to reconcile kinetic equations with equations of microdynamics. We postulate not the equations of microdynamics, but the kinetic equations and treat their microscopic solutions as the particle dynamics. The establishment of the existence of microscopic solutions can be regarded as a kind of derivation of the particle dynamics from a kinetic equation. So, instead of the derivation of kinetic equations from the microdynamics we suggest a derivation of the microdynamics from kinetic equations.

In this interpretation the ``fundamental'' laws of motion are not the laws of motion for individual particles (reductionism), but the laws of motion for the single-particle density function, which describes the system as a whole (holism). Laws of motion for the individual particles are approximations.

The Boltzmann--Enskog equation describes the dynamics of a hard sphere gas. In the final section of the paper we consider other kinetic equations. Namely, we consider the Boltzmann--Enskog equation with an additional ``Vlasov integral'' (we call it the Boltzmann--Enskog--Vlasov equation) and the Boltzmann equation. The Boltzmann--Enskog--Vlasov equation describes the dynamics of hard spheres that additionally interact with some long-ranged interaction potential. This equation also has microscopic solutions and, hence, one can speak about a derivation of the Newton equation (or Hamilton equations) of microdynamics from it. 

The Boltzmann equation has not microscopic solutions, but it can be regarded as an approximation of the Boltzmann--Enskog (or a Boltzmann--Enskog-like) equation in the case of a large number of particles and a negligible interaction radius of them (the Boltzmann--Grad limit). In this sense one can speak about a reconciliation of the microdynamics with the Boltzmann equation as well.

Recently Igor Volovich proposed functional mechanics \cite{VolFuncMech,VolRand} (see also \cite{Mikh,PiskVol,Pisk,TrVolFuncRat}) for solving the irreversibility problem. The functional mechanics states that the fundamental notion of mechanics is not a material point, but an (integrable) probability density function. Hence, the fundamental equation even for a single particle is not the Newton equation (or Hamilton equations), but the Liouville equation (or even the Fokker--Planck--Kolmogorov equation). A solution of the Newton equation can be regarded as a solution of the Liouville equation in the form of a delta-function (another interpretation: the Hamilton equations are the equations of characteristics for the Liouville equation). But delta-functions are not allowed to be considered. A particle with definite position and momentum is an approximation, which is suitable if the integrable density function has a specific (delta-like) form.

As we have mentioned above, the spatial density for an integrable solution of the Liouville equation in a box weekly converges to the uniform distribution. So, if we postulate not the Newton, but the Liouville equation for a single particle, we obtain a kind of relaxation to the equilibrium state for it. Thus, the paradox between the recurrence of the microdynamics and the convergence to the equilibrium state of the macrodynamics is eliminated: both dynamics exhibit the convergence to the equilibrium.

We use the same idea: we notice that the solution of the Newton equation for the system of hard spheres can be regarded as a special type solution (in the form of the sum of delta-functions) of the Boltzmann--Enskog equation. Then we postulate the Boltzmann--Enskog equation. The existence of microscopic solutions allows us to treat the particle dynamics as an approximation, which is suitable if the density function has a specific (delta-like) form.

Also Igor Volovich obtained a kinetic equation for a finite number of particles \cite{VolBogol}. In fact, the Boltzmann--Enskog equation is also a kinetic equation for a finite number of particles.

I am very happy to dedicate this report to the 65th anniversary of Professor Igor Volovich, and I would like to wish him good health and further scientific successes.

\section{Boltzmann--Enskog equation and its microscopic solutions}

Let us consider a gas of hard spheres with diameter $a>0$ in some 3-dimensional domain. In order to avoid difficulties with taking a boundary of the domain into account, let us assume that the domain is the 3-dimensional torus $\mathbb T^3$. The dynamics in a rectangular box is known to be reduced to the dynamics on a torus. The Boltzmann--Enskog kinetic equation for this gas takes the form

\begin{equation}\label{EqBE}
\frac{\partial f(r_1,v_1,t)}{\partial t}=-v_1\frac{\partial f(r_1,v_1,t)}{\partial r_1}+\St f,
\end{equation}
where
\begin{equation}\label{EqStBE}
\St f=na^2\int_{(v_{21},\sigma)\geq0}
(v_{21},\sigma)[f(r_1,v'_1,t)f(r_1+a\sigma,v'_2,t)-f(r_1,v_1,t)f(r_1-a\sigma,v_1,t)]
d\sigma dv_2.
\end{equation}
Here $f(r_1,v_1,t)$ is a density of particles near the point $r_1\in\mathbb T^3$ that have velocities near $v_1\in\mathbb R^3$ at the moment of time $t$. Let $f$ be normalized on the volume $V$ of the torus $\mathbb T^3$: 

$$\int_{\mathbb T^3\times\mathbb R^3}f(r_1,v_1,t)dr_1dv_1=V.$$

The expression $\St f_1$ is called the collision integral. $n>0$ is a constant (the mean concentration of particles), $v_1$ and $v_2$ are the velocities of the colliding particles just before the collision, $v'_1$ and $v'_2$ are the velocities of the particles just after the collision:

\begin{equation}\label{EqHardColl}
\begin{aligned}
v'_1&=v_1+\sigma(v_{21},\sigma),\\
v'_2&=v_2-\sigma(v_{21},\sigma).
\end{aligned}
\end{equation}
Here $v_{21}=v_2-v_1$, the vector $\sigma$ is a unit vector from the center of the second particle to the center of the first particle (so, $\sigma$ belongs to the 2-dimensional unit sphere: $\sigma\in S^2$), $(\cdot,\cdot)$ is a scalar product.

The Boltzmann--Enskog equation describes the irreversible dynamics of a hard sphere gas and the entropy production, a rigorous $H$-theorem for this equation (and for the more general Enskog-type equations) can be found in \cite{Resib,BelLach91}. The $H$-function for the Boltzmann--Enskog equation slightly differs from the $H$-function for the Boltzmann equation.

Now let us consider a function $f$ which has the form of a sum of delta-functions:
\begin{equation}\label{Eqfdelta}
f(r_1,v_1,t;\Gamma)=n^{-1}\sum_{i=1}^N\delta(r_1-q_i(t,\Gamma),v_1-w_i(t,\Gamma)),
\end{equation}
where $q_1(t,\Gamma),w_1(t,\Gamma),\ldots, q_N(t,\Gamma),w_N(t,\Gamma)$ are the positions and velocities of $N$ hard spheres at the moment $t$. Here $n=\frac NV$. The time dependence is defined by the initial positions and velocities 
$$\Gamma=(q_1^0,w_1^0,\ldots, q_N^0,w_N^0)$$ 
and the laws of motion for hard spheres (i.e., law (\ref{EqHardColl}) for pairwise collisions and the free motion between the collisions; we neglect the collisions of three and more spheres simultaneously, because the corresponding initial phase points $\Gamma$ have zero measure as well as initial phase points which lead to the infinite number of collisions in a finite time \cite{CPG}).

\begin{proposition}[Bogolyubov]
The (generalized) function $f$ given by (\ref{Eqfdelta}) satisfies the Boltzmann--Enskog equation (\ref{EqBE}).
\end{proposition}

A solution of form (\ref{Eqfdelta}) is called a microscopic solution of a kinetic equation. This result is surprising for two reasons. The first reason is that the Boltzmann--Enskog is regarded as an approximation to the $N$-particle dynamics in the case of a dilute gas ($n$ is small). This result shows that this equation has exact microscopic solutions as well. The Vlasov kinetic equation is also known to have microscopic solutions \cite{Vlasov}.

But, in contrast to the Vlasov equation, the Boltzmann--Enskog equation is time-irreversible and describes the entropy production. This proposition tells that it has solutions $f$, which are defined by the time-reversible dynamics of $N$ hard spheres. In other words, both $f(r_1,v_1,t;\Gamma)$ and $\widetilde f(r_1,v_1,t;\Gamma)=f(r_1,-v_1,-t;\Gamma)$ are solutions of the Boltzmann--Enskog equation, despite of the fact that this equation is time-irreversible! 

Indeed, $f(r_1,-v_1,-t;\Gamma)=f(r_1,v_1,t;\widetilde\Gamma)$, where $\widetilde\Gamma=(q_1^0,-w_1^0,\ldots,q_N^0,-w_N^0)$. So, the function $\widetilde f$ also has form (\ref{Eqfdelta}) and, hence, is a solution of (\ref{EqBE}).

\section{Rigorous formulation of the microscopic solutions of the Boltzmann--Enskog equation}

Let us explore the relation between the irreversible Boltzmann--Enskog equation (\ref{EqBE}) and its reversible microscopic solutions (\ref{Eqfdelta}). In this section we define solutions $f$ of form (\ref{Eqfdelta}) rigorously. 

Denote $G=\mathbb T^3\times\mathbb R^3$ (the single-particle phase space), $\Omega$ is the set of points $(r_1,v_1,\ldots,r_N,v_N)\in G^N$ such that $|r_i-r_j|\geq a$ for all $i\neq j$. Obviously, $\Gamma\in\Omega$, since hard sphere centers cannot be closer than $a$ to each other. Let $D(\Gamma)$ be a continuously differentiable function on $\Omega$ and rapidly decrease in $v_1,\ldots,v_N$ with all its first-order partial derivatives in $r$ and $v$. Let also the following condition be satisfied:
\begin{equation}\label{EqCollCont}
D(\ldots,r_i,v_i,\ldots,r_i-a\sigma,v_j,\ldots)=D(\ldots,r_i,v'_i,\ldots,r_i-a\sigma,v'_j,\ldots)
\end{equation}
for all $i\neq j,r_i,v_i,v_j,\sigma$ and other arguments (denoted as ``\ldots'') such that $(v_j-v_i,\sigma)\geq0$. Here $(r_i,v_i)$ and $(r_i-a\sigma,v_j)$ are $i$th and $j$th pairs of arguments of $D$, the other arguments are the same in both sides of the equalities, $v'_i$ and $v'_j$ are defined by (\ref{EqHardColl}) (with indexes 1 and 2 replaced by $i$ and $j$). Note that we can regard $D$ as a function on $G^N$ with $D=0$ if $|r_i-r_j|<a$ for some $i\neq j$.

Put by definition
\begin{equation}\label{EqProdDelta}
\bigl(\prod_{i=1}^N\delta(x_i-y_i(t,\Gamma)),D(\Gamma)\bigr)=S^{(N)}_{-t}D(x_1,\ldots,x_N)
\end{equation}
where $x_i=(r_i,v_i)$, $y_i(t)=(q_i(t),w_i(t))$. Here $S^{(N)}_t$ is the dynamic flow, i.e., $S^{(N)}_t(x_1,\ldots,x_N)=(X_1(t),\ldots,X_N(t))$, where $X_i(t)=(R_i(t),V_i(t))\in\mathbb T^3\times\mathbb R^3$ are the positions and velocities of $N$ hard spheres at the moment $t$ if their positions and momenta at the moment $t=0$ were $(x_1,\ldots,x_N)$. We use the notation 
$$S^{(N)}_tD(x_1,\ldots,x_N)\equiv D(S^{(N)}_t(x_1,\ldots,x_N)).$$
Here $S^{(N)}_t$ can be regarded as an operator which maps the function $D$ to the function with the shifted arguments.

The meaning of definition (\ref{EqProdDelta}) is the following. It is natural to define the expression $(\prod_{i=1}^N\delta(x_i-y_i(t,\Gamma)),D(\Gamma))$ as the value of $D$ in a point $\Gamma$ such that $y_i(t,\Gamma)=x_i$ for all $i=1,\ldots,N$. This is exactly $S^{(N)}_{-t}D(x_1,\ldots,x_N)$. 

If $D$ has the meaning of the initial probability density function, then $S^{(N)}_{-t}D\equiv D_t$ has the meaning of the probability density function at the time $t$.

Further,
$$\bigl(\delta(x_i-y_i(t,\Gamma)),D(\Gamma)\bigr)=
\int_{G^{N-1}} S^{(N)}_{-t}D(x_1,\ldots,x_N)\,dx_1\ldots dx_{i-1}dx_{i+1}\ldots dx_N.$$
So, we define the left-hand side of the expression as the integral from $D(\Gamma)$ over the initial values $\Gamma$ such that $y_i(t,\Gamma)=x_i$. This definition can be obtained by the formal integration of (\ref{EqProdDelta}) over all variables $x_1,\ldots,x_N$, except $x_i$. Analogously,
$$\bigl(\delta(x_i-y_i(t,\Gamma))\delta(x_j-y_j(t,\Gamma)),D\bigr)=
\int_{G^{N-2}} S^{(N)}_{-t}D\,dx_1\ldots dx_{i-1}dx_{i+1}\ldots dx_{j-1}dx_{j+1}\ldots dx_N.$$

Denote

\begin{equation}\label{EqFfD}\begin{aligned}
(f(x_1,t;\Gamma),D(\Gamma))&=F_1(x_1,t),\\
(f(x_1,t;\Gamma)f(x_2,t;\Gamma),D(\Gamma))&=F_2(x_1,x_2,t).
\end{aligned}\end{equation}
$F_1$ and $F_2$ are the single-particle and two-particle density functions correspondingly. If $D(x_1,\ldots,x_N)$ is symmetric relative to the permutations of $x_i$ (the particles are indistinguishable), then
\begin{equation}\label{EqF}\begin{aligned}
F_1(x_1,t)&=V\int_{G^{N-1}}S^{(N)}_{-t}D(x_1,\ldots,x_N)\,dx_2\ldots dx_N,\\
F_2(x_1,x_2,t)&=V\bigl(1-\frac 1N\bigr)\int_{G^{N-2}} S^{(N)}_{-t}D(x_1,\ldots,x_N)\,dx_3\ldots dx_N.
\end{aligned}\end{equation}

Condition (\ref{EqCollCont}) can be rewritten in terms of the two-particle density function $F_2$:
\begin{equation}\label{EqCollContF2}
F_2(r_1,v_1,r_1-a\sigma,v_2,t)=F_2(r_1,v'_1,r_1-a\sigma,v'_2,t)
\end{equation}
for all $r_1,v_1,v_2,\sigma$ and $t$ such that $(v_{21},\sigma)\geq0$.

Note that the operator $S^{(N)}_t$ is discontinuous, since the velocities change jump-like at the moments of collisions. But if condition (\ref{EqCollCont}) is satisfied, then $S^{(N)}_{-t-0}D=S^{(N)}_{-t+0}D$.

It would be desirable to define $f$ as a generalized function on some space of test functions. We can define $f$ as a functional on $\mathcal D(\Omega\backslash\partial\Omega)$ (infinitely differentiable functions with compact support in $\Omega\backslash\partial\Omega$ where $\partial\Omega$ is a boundary of $\Omega$). Equality (\ref{EqCollCont}) is trivially satisfied: both sides of it are equal to zero, since the points with $|r_i-r_j|=a$ for some $i\neq j$ belong to $\partial\Omega$. But the conditions imposed upon a function from $\mathcal D(\Omega\backslash\partial\Omega)$ (infinite differentiability and equality to zero on the boundary) are too restrictive, we want $f$ to be defined on a wider class of functions.

We can define $f$ as a functional on the space of continuously differentiable functions on $\Omega$ that rapidly decrease in $v_1,\ldots,v_N$. Denote this space as $SC^1(\Omega)$. A topology on $SC^1(\Omega)$ can be specified by the following seminorms:
$$P_K(D)=\max\{\sup_{\Omega}(1+v^2)^{\frac K2}|D|,\:\sup_{\Omega}(1+v^2)^{\frac K2}\left|\frac{\partial D}{\partial r_i}\right|,\:\sup_{\Omega}(1+v^2)^{\frac K2}\left|\frac{\partial D}{\partial v_i}\right|,\:i=1,2,\ldots,N\},$$
where $D\in SC^1(\Omega)$, $v^2=v_1^2+\ldots+v_N^2$, and $K=1,2,\ldots$. But in this case, we must impose additional functional condition (\ref{EqCollCont}). Let us regard $f$ as a functional (generalized function) on the subspace of $SC^1(\Omega)$ defined by additional condition (\ref{EqCollCont}). Denote this subspace as $CC^1(\Omega)$.

The meaning of condition (\ref{EqCollCont}) (or (\ref{EqCollContF2})) is easy to understand if $D$ is a probability density function. This is a continuity condition: since the collision of hard spheres is instantaneous, the two-particle density just after the collision should be equal to the density just before the collision. If this condition is not satisfied, then the function $D_t$ is, generally speaking, discontinuous in $x_1,\ldots,x_N$. So, we demand the continuity not only for the initial density function, but for the density function at every moment of time as well. Also, this condition means the continuity of $D_t$ in $t$.

In fact, we can drop condition (\ref{EqCollCont}). But this case is more sophisticated to analyze. For simplicity we will assume condition (\ref{EqCollCont}).

Now we are able to formulate proposition~1 as a rigorous assertion.

\begin{proposition}
The (generalized) function $f$ given by (\ref{Eqfdelta}) satisfies the Boltzmann--Enskog equation (\ref{EqBE}). The both sides of the equation are understood as functionals on $CC^1(\Omega)$.
\end{proposition}

\begin{proof}
The single-particle and two-particle density functions for the system of $N$ hard spheres are known to satisfy the following equation (the first equation of the BBGKY hierarchy for hard spheres, see \cite{Cerc}; condition (\ref{EqCollCont}) is used in the derivation of this equation):
\begin{multline}\label{EqBEB}
\frac{\partial F_1(r_1,v_1,t)}{\partial t}=-v_1\frac{\partial F_1(r_1,v_1,t)}{\partial r_1}+na^2\int_{(v_{21},\sigma)\geq0}
(v_{21},\sigma)[F_2(r_1,v'_1,r_1+a\sigma,v'_2,t)\\-F_2(r_1,v_1,r_1-a\sigma,v_1,t)]
d\sigma dv_2,
\end{multline}
where $n=\frac NV$. Here the argument $r_1\pm a\sigma$ means $r_1\pm a\sigma\pm0$, since $F_2(r_1,v_1,r_2,v_2,t)=0$ as $|r_1-r_2|<a$. Due to (\ref{EqFfD}) this equation can be written as
$$\int\bigl\{\frac{\partial f(r_1,v_1,t;\Gamma)}{\partial t}+v_1\frac{\partial f(r_1,v_1,t;\Gamma)}{\partial r_1}-\St f\bigr\}D(\Gamma)\,d\Gamma=0,$$
where $\St f$ is given by (\ref{EqStBE}). Since $D(\Gamma)$ is an arbitrary function from $CC^1(\Omega)$, this equality means that $f$ satisfies the Boltzmann--Enskog equation in the stated sense.
\end{proof}

It should be noted that the microscopic solutions of the Boltzmann--Enskog equation are understood slightly differently than the microscopic solutions of the Vlasov equation. In the case of the Vlasov equation, a generalized function in form (\ref{Eqfdelta}) is understood as a functional on $\mathcal D(G)$ (infinitely differentiable functions on $G$ with compact support). So, if $A\in\mathcal D(G)$, then
$$(\delta(r-q_i(t,\Gamma),v-w_i(t,\Gamma)),A(r,v))=A(q_i(t,\Gamma),w_i(t,\Gamma)).$$
Thus, test functions depend on $r$ and $v$. But it is difficult to give a sense to the expressions $f(r_1,v_1,t)f(r_1-a\sigma,v_2,t)$ and $\frac{df}{dt}$ in the Boltzmann--Enskog equation as functionals on $\mathcal D(G)$. For this reason, they are defined on $CC^1(\Omega)$. In particular, test functions depend on $\Gamma\in\Omega$. In this case, all terms can be well-defined.

It could be seemed that the microscopic solutions of the Boltzmann--Enskog equation are just a paraphrase of equation (\ref{EqBEB}). However, the following theorem reveal another view on the microscopic solutions.

\section{Approximate microscopic solutions}

\begin{theorem}
Consider the function
\begin{equation*}
f^0_\varepsilon(r_1,v_1)=n^{-1}\sum_{i=1}^N\delta^0_\varepsilon(r_1-q_i^0,v_1-w_i^0),\end{equation*}
where $\delta^0_\varepsilon(r,v)\in  C^1_0(G)$ (i.e., $\delta^0_\varepsilon$ is a continuously differentiable function in $G$ with compact support), $\delta^0_\varepsilon(r,v)\to\delta(r,v)$ and the support of $\delta^0_\varepsilon$ tends to $\{0,0\}$ as $\varepsilon\to0$. The generalized function $\delta(r,v)$ is defined on the space $\mathcal D(G)$ of test functions. Let $|q_i^0-q_j^0|>a$ for all $i\neq j$. 

1) Then, for every $T>0$ there exist a number $\varepsilon_T>0$ and a family of functions $D_\varepsilon\in CC^1(\Omega)$, $\varepsilon>0$, such that there exists a solution $f_\varepsilon(r_1,v_1,t)\in C^1(G\times[-T,T])$ of the Cauchy problem for the Boltzmann--Enskog equation with the initial data $f^0_\varepsilon$. This solution coincides with $F_{1,\varepsilon}(r_1,v_1,t)$ on $(r_1,v_1,t)\in G\times[-T,T]$ for all $\varepsilon<\varepsilon_T$. Here $F_{1,\varepsilon}$ is defined by the first formula of (\ref{EqF}) with $D$ substituted by $D_\varepsilon$.

2) Consider an arbitrary time moment $t\in[-T,T]$ such that $|q_i(t,\Gamma)-q_j(t,\Gamma)|>a$ for all $i\neq j$. Here $q_i(t,\Gamma)$, $w_i(t,\Gamma)$, $i=1,\ldots,N$, and $\Gamma$ are defined as before. Then,
$$\lim_{\varepsilon\to0}f_\varepsilon(r_1,v_1,t)=n^{-1}\sum_{i=1}^N\delta(r_1-q_i(t,\Gamma), v_1-w_i(t,\Gamma)).$$
\end{theorem}
\begin{proof}
Consider the density function

$$D_\varepsilon(\Lambda)=\frac1{N!}\sum_{(j_1,\ldots,j_N)}\prod_{i=1}^N\delta^0_\varepsilon(\xi_i^0-q_{j_i}^0, u_i^0-w_{j_i}^0),$$
where $\Lambda=(\xi^0_1,u^0_1,\ldots,\xi^0_N,u^0_N)\in G^N$, the sum is taken over all permutations $j_1,\ldots,j_N$ of $1,\ldots,N$. Due to the condition $|q_i^0-q_j^0|>a$ for all $i\neq j$ and the conditions imposed upon $\delta^0_\varepsilon$, for a given $T$, there exists $\varepsilon_0>0$ such that $D_\varepsilon(\Lambda)=0$ whenever $|\xi^0_i-\xi^0_j|\leq a$ for some $i\neq j$, for all $\varepsilon<\varepsilon_0$. Hence, we can treat $D_\varepsilon$ as a function on $\Omega\ni\Lambda$. In this case, condition (\ref{EqCollCont}) is also satisfied, since both sides of it are equal to zero. Note that by the same reason, condition (\ref{EqCollCont}) is also satisfied for the derivatives of $D_\varepsilon$ as well. Hence, $D_{t,\varepsilon}\equiv S^{(N)}_{-t}D_\varepsilon$ also belongs to $C^1(\Omega)$ for all $t$ and is continuously differentiable in $t$ as well.

Obviously, $F^0_{1,\varepsilon}(r_1,v_1)\equiv f^0_\varepsilon(r_1,v_1)$, where $F^0_{1,\varepsilon}$ is defined by the first formula in (\ref{EqF}) for $t=0$ and $D$ substituted by $D_\varepsilon$. 

For an arbitrary $t\in\mathbb R$, we have
$$D_{t,\varepsilon}(\Lambda)\equiv S^{(N)}_{-t}D_\varepsilon(\Lambda)=\frac1{N!}\sum_{(j_1,\ldots,j_N)}\prod_{i=1}^N\delta^0_\varepsilon(\xi_i(-t,\Lambda)-q_{j_i}^0, u_i(-t,\Lambda)-w_{j_i}^0),$$
where $\xi_1(-t,\Lambda),u_1(-t,\Lambda),\ldots,\xi_N(-t,\Lambda),u_N(-t,\Lambda)$ are the positions and momenta of $N$ hard spheres whenever their initial positions and momenta were $\Lambda$. There exists $\varepsilon_1>0$ (let $\varepsilon_1\leq\varepsilon_0$) such that the last formula can be expressed as
$$D_{t,\varepsilon}(\Lambda)=\frac1{N!}\sum_{(j_1,\ldots,j_N)}\prod_{i=1}^N\delta^{j_i}_\varepsilon(\xi_i^0-q_{j_i}(t,\Gamma), u_i^0-w_{j_i}(t,\Gamma),t,\Gamma)$$
whenever $\varepsilon<\varepsilon_1$. Here $\delta^i_\varepsilon(r,v,t,\Gamma)\to\delta(r,v)$ as $\varepsilon\to0$ for all $i=1,\ldots,N$ and for all $t$ such that $|q_i(t,\Gamma)-q_j(t,\Gamma)|>a$, $i\neq j$. Intuitively: if $\varepsilon$ is sufficiently small and the points $\Lambda$ with nonzero $D(\Lambda)$ are concentrated near $\Gamma$, it is almost equivalent to shift $\Lambda$ on $-t$ along the dynamical flow or to shift $\Gamma$ on $t$ along the same flow. For a given $T>0$, there exists $\varepsilon_T>0$ (let $\varepsilon_T\leq\varepsilon_1$) such that the support of $\delta^i_\varepsilon(r,v,t)$ lies in $B_a\times\mathbb R^3$ for all $i$, $t\in[-T,T]$, and $\varepsilon<\varepsilon_T$. Here $B_a=\{(r,v)\in G:\,|r|<a\}$. 

The functions $F_1$ and $F_2$ given by (\ref{EqF}) with $D$ substituted by $D_\varepsilon$ have the forms
$$F_{1,\varepsilon}(x_1,t)=n^{-1}\sum_{i=1}^N\delta^i_\varepsilon(r_1-q_i(t,\Gamma),v_1-w_i(t,\Gamma),t,\Gamma),$$
\begin{equation*}\begin{split}F_{2,\varepsilon}(x_1,x_2,t)=n^{-2}\bigl(1-\frac1N\bigr)^{-1}\sum_{i\neq j}
&\delta^i_\varepsilon(r_1-q_i(t,\Gamma), v_1-w_i(t,\Gamma),t,\Gamma)\\\times&\delta^j_\varepsilon(r_2-q_j(t,\Gamma), v_2-w_j(t,\Gamma),t,\Gamma).\end{split}\end{equation*}

If $\varepsilon<\varepsilon_T$, then \begin{equation}\label{EqChaos}F_{2,\varepsilon}(x_1,x_2,t)=\begin{cases}F_{1,\varepsilon}(x_1,t)F_{1,\varepsilon}(x_2,t),&\text{if } |r_2-r_1|\geq a,\\0,&\text{otherwise}\end{cases}\end{equation}
for all $t\in[-T,T]$.

The functions $F_1$ and $F_2$ satisfy equation (\ref{EqBE}). But in view of (\ref{EqChaos}), equation (\ref{EqBE}) is reduced to the Boltzmann--Enskog equation for $F_{1,\varepsilon}$. Remind that $F_{1,\varepsilon}(x_1,0)\equiv F^0_{1,\varepsilon}(x_1)=f^0_\varepsilon(x_1)$. Hence, $F_{1,\varepsilon}(x_1,t)$ is a solution of the Cauchy problem for the Boltzmann--Enskog equation with the initial data $f^0_\varepsilon$ whenever $\varepsilon<\varepsilon_T$. $F_{1,\varepsilon}\in C^1(G\times[-T,T])$ in force of the continuous differentiability of $D_{t,\varepsilon}(\Lambda)$ in $t$ and $\Lambda$. The assertion of paragraph 1) of the theorem has been proved. The assertion of paragraph 2) immediately follows from the facts that $f_\varepsilon\equiv F_{1,\varepsilon}$ and $\delta^i_\varepsilon(r,v,t,\Gamma)\to\delta(r,v)$ as $\varepsilon\to0$ for all $i=1,\ldots,N$ and for all $t$ such that $|q_i(t,\Gamma)-q_j(t,\Gamma)|>a$, $i\neq j$.\end{proof}

So, if the density function in the Boltzmann--Enskog equation has the form of a sum of continuous delta-like distributions, then these distributions behave as ``spreaded hard spheres'': they collide with each other, move freely between the collisions, and gradually spread. The Boltzmann--Enskog equation (\ref{EqBE}) coincide with the first equation of the BBGKY hierarchy (\ref{EqBEB}). When the spatial spreading of the hard sphere centers exceeds the hard sphere diameter, equations (\ref{EqBE}) and (\ref{EqBEB}) cease to coincide. Note that this moment $T$ increases as $\varepsilon$ decreases. If $\varepsilon=0$ and $f_\varepsilon$ is a sum of delta-functions (\ref{Eqfdelta}), then $T=\infty$. This is the case of propositions~1 and~2. So, the approximate microscopic solutions turns to the exact microscopic solutions in the limit.

The considered solutions $f_\varepsilon$ can be called the ``approximate microscopic solutions''. Like the exact microscopic solutions, they are reversible, but only on the bounded time segments.

So, the microscopic solutions of the Boltzmann--Enskog equation are not just a paraphrase of another equation (\ref{EqBEB}). Undoubtedly, the microscopic solutions reflect some fundamental properties of the Boltzmann--Enskog equation. They have the following meaning: they describe the limiting motion of the centers of delta-like distributions in the approximate microscopic solutions.

Note also that equality (\ref{EqChaos}) is an analogue of the molecular chaos condition (or Stoss\-zahl\-an\-satz) for the Boltzmann--Enskog equation.

\section{Reversible and irreversible classes of solutions of the Boltzmann--Enskog equation}

Now we are able to answer the question about the relation between the irreversible Boltzmann--Enskog equation (\ref{EqBE}) and its reversible solutions (\ref{Eqfdelta}). For generality let us denote $\Delta$ the generalized functions in form (\ref{Eqfdelta}) with arbitrary functions $q_i(t,\Gamma)$, $w_i(t,\Gamma)$, $i=1,\ldots,N$. 
\begin{theorem}
1) Let us treat the Boltzmann--Enskog equation (\ref{EqBE}) as an equation for functions $f\in C^1(G\times[0,T))$, $T>0$ (the classical solution of the Boltzmann--Enskog equation is defined in this class or its subclass \cite{Polew,HaNoh}). In this case, the equation is irreversible, i.e., if some function $f(r_1,v_1,t)$ is a solution of this equation, then, in general, the function $\widetilde f(r_1,v_1,t)=f(r_1,-v_1,-t)$ is not a solution;

2) Let us treat the Boltzmann--Enskog equation as an equation for generalized functions $f\in\Delta$ on test functions from $CC^1(\Omega)$. In this case, the equation is reversible, i.e., if some generalized function $f(r_1,v_1,t;\Gamma)$ is a solution of this equation, then the generalized function $\widetilde f(r_1,v_1,t;\Gamma)=f(r_1,-v_1,-t;\Gamma)$ is also a solution.

3) The continuously differentiable solutions $f_\varepsilon(r_1,v_1,t)$ constructed in theorem~1 (sums of delta-like functions) are reversible on time segments $[-T,T]$, where $T>0$ depends on $\varepsilon>0$.
\end{theorem}

\begin{proof}
Since $\frac{\partial\widetilde f(r_1,v_1,t)}{\partial t}=\frac{\partial f(r_1,-v_1,-t)}{\partial t}=
-\frac{\partial f(r_1,-v_1,-t)}{\partial(-t)}$, we have

\begin{equation*}\begin{split}
\frac{\partial\widetilde f(r_1,v_1,t)}{\partial t}=
-v_1\frac{\partial f(r_1,-v_1,-t)}{\partial r_1}-a^2\int_{(-v_{21},\sigma)\geq0}&|(v_{21},\sigma)|
[f(r_1,-v'_1,-t)f(r_1+a\sigma,-v'_2,-t)\\&-f(r_1,-v_1,-t)f(r_1-a\sigma,-v_2,-t)]
d\sigma dv_2.
\end{split}
\end{equation*}
Replace $f(r,-v,-t)=\widetilde f(r,v,t)$ and change the variable $\sigma\to-\sigma$:
\begin{equation*}\begin{split}
\frac{\partial\widetilde f(r_1,v_1,t)}{\partial t}=-v_1\frac{\partial\widetilde f(r_1,v_1,t)}{\partial r_1}+
a^2\int_{(v_{21},\sigma)\geq0}(v_{21},\sigma)
[&\widetilde f(r_1,v_1,t)\widetilde f(r_1+a\sigma,v_2,t)\\-&\widetilde f(r_1,v'_1,t)\widetilde f(r_1-a\sigma,v'_2,t)]d\sigma dv_2.
\end{split}\end{equation*}
The Boltzmann--Enskog equation is obtained from this equation by swapping the velocities $v_{1,2}$ and $v'_{1,2}$ in the collision integral. Hence, the function $\widetilde f$ satisfies the Boltzmann--Enskog equation if and only if
\begin{equation*}
\widetilde f(r_1,v_1,t)\widetilde f(r_1-a\sigma,v_2,t)=\widetilde f(r_1,v'_1,t)\widetilde f(r_1-a\sigma,v'_2,t),
\end{equation*}
or, equivalently,
\begin{equation}\label{EqCollContf}
f(r_1,v_1,t)f(r_1-a\sigma,v_2,t)=f(r_1,v'_1,t)f(r_1-a\sigma,v'_2,t),
\end{equation}
for all $r_1$, $v_1$, $\sigma$ and $t$.

Condition (\ref{EqCollContf}) is, in general, not satisfied if $f$ is a classical solution. Hence, if solutions from $C^1(G\times[0,T))$ are considered, the Boltzmann--Enskog equation is irreversible.

Now let $f\in\Delta$. In this case, condition (\ref{EqCollContf}) is satisfied, since it is equivalent to (\ref{EqCollContF2}) by definition. Indeed, both sides of (\ref{EqCollContf}) are understood as functionals on $CC^1(\Omega)$:
\begin{equation}\label{EqCollContfD}
\int f(r_1,v_1,t;\Gamma)f(r_1-a\sigma,v_2,t;\Gamma)D(\Gamma)\,d\Gamma=\int f(r_1,v'_1,t;\Gamma)f(r_1-a\sigma,v'_2,t;\Gamma)D(\Gamma)\,d\Gamma.
\end{equation}
This is exactly condition (\ref{EqCollContF2}). Hence, if solutions from $\Delta$ are considered, the Boltzmann--Enskog equation is reversible.

Finally, if we consider the solutions $f_\varepsilon(r_1,v_1,t)$ constructed in theorem~1, then condition (\ref{EqCollContf}) is also satisfied if $t\in[-T,T]$. Truly, it coincides with (\ref{EqCollContF2}) (with $F_2$ replaced by $F_{2,\varepsilon}$), since $f_\varepsilon(x_1,t)f_\varepsilon(x_2,t)=F_{2,\varepsilon}(x_1,x_2,t)$ for $t\in[-T,T]$.
\end{proof}

So, condition (\ref{EqCollContF2}) for test functions is crucial for the reversibility of the Boltzmann--Enskog equation for the class of solutions $\Delta$. If we drop this condition, then, in general, the microscopic solution (\ref{Eqfdelta}) does not satisfy the Boltzmann--Enskog equation. The right-hand side and left-hand side time derivatives of $f$ are different, only the right-hand side time derivative is equal to the right-hand side of the Boltzmann--Enskog equation. Using a more sophisticated analysis we can prove the assertion of theorem~1 for this case as well.

This result can be compared with the Kozlov's diffusion theorems for a collisionless continuous medium in a rectangular box \cite{Kozlov}. The dynamics is described by the Liouville equation. If we consider a solution in the form of a delta-function, then the solution is periodic or almost periodic (since this case is reduced to the single-particle dynamics). But if we consider an integrable solution, the spatial density weakly converges to the uniform distribution, as time indefinitely increases. So, the asymptotic dynamical properties of solutions depend on the considered class of functions.

Here we show that in the case of the Boltzmann--Enskog equation, not only the asymptotic properties of solutions, but also the reversibility or irreversibility property of them depend on the considered class of functions.

\section{Reconciliation of the kinetic equations with the microscopic dynamics}

Of course, the above arguments cannot be regarded as a derivation of the kinetic Boltzmann--Enskog equation from the microdynamics. According to the traditional point of view, we must start with some initial $N$-particle distribution function $D$ (a randomization of the initial $N$-particle phase point) and analyse the dynamics $D_t=S^{(N)}_{-t}D$ of it. By doing this, we get equation (\ref{EqBEB}) (as well as the other equations of the BBGKY hierarchy). Then we must obtain a kinetic equation using some limiting procedure. But, as we have stated in the introduction, until now we have not a fully satisfactory realization of this program.

However, the microscopic solutions of the Boltzmann--Enskog equation and the investigations of the reversibility or irreversibility property of it suggest us another way of the reconciliation of the kinetic equations with the microdynamics.

Let us postulate not the microdynamics, but the Boltzmann--Enskog equation. That is, the dynamics of a system that we call ``a system of hard spheres'' is postulated to be described by the Boltzmann--Enskog equation (\ref{EqBE}) with a \textit{continuous} function $f$. Then, this equation is noticed to have also special (noncontinuous) solutions in form (\ref{Eqfdelta}) (or the continuous solutions constructed in theorem~1). The dynamics of $(q_i(t,\Gamma),w_i(t,\Gamma))$, $i=1,\ldots,N$, in these solutions is interpreted as the microscopic dynamics. 

So, the establishment of the existence of microscopic solutions for a kinetic equation can be regarded as a derivation of the particle dynamics from it. Instead of the derivation of the kinetic equations from the microdynamics we suggest a kind of derivation of the microdynamics from the kinetic equations.

According to this point of view, the ``fundamental'' laws of motion are not the laws of motion for individual particles (reductionism), but the laws of motion for the function $f$, which describes the system as a whole (holism). Laws of motions for the individual particles are approximations. This approximation is suitable if the density function has a special (delta-like) form (see, e.g., theorem~1).

Functional mechanics \cite{VolFuncMech,VolRand,Mikh,PiskVol,Pisk,TrVolFuncRat} also claims that the Newton equation (or Hamilton equations) is not fundamental even in the classical nonrelativistic world.

Also, it should be noted that, in principle, the function $f$ in the Boltzmann--Enskog equation, as well as other kinetic equations, can be interpreted without any reference to the notion of a particle. Let $f$ be normalized not on the volume $V$, but on the total mass $M=Nm$ ($m$ is the mass of a particle): $\int f(r,v,t)drdv=M$. Then the Boltzmann--Enskog equation has the same form, but the constant $n=\frac NV$ in the collision integral (\ref{EqStBE}) is replaced by $\frac1m$. $f(r,v,t)\Delta r\Delta v$ has the meaning of the mass of the matter which is contained in the volume element $\Delta r$ and moves with the velocity in the range $\Delta v$. So, $f(r,v,t)$ is the density function of the matter in the phase point $(r,v)\in\mathbb T^3\times\mathbb R^3$.

We have a kind of continuous medium. In contrast to the usual continuum mechanics, there is no definite velocity in a fixed point of physical space: in the same volume element $\Delta r$ different parts of the matter move with different velocities. The terms $f(r_1,v_1,t)f(r_1\pm a\sigma,v_2,t)$ in the collision integral are intensities of interaction of different parts of the matter.

Initially, atomic theory of matter was an attempt to reduce the properties of the matter to the properties of its elementary blocks. But the problems like the reversibility paradox show that this program and the notion of an atom should be revised.

Thus, we may still interpret the function $f$ as the single-particle density function, or we may interpret it as the density function of a continuous matter. In both cases, we postulate the Boltzmann--Enskog equation for the system as a whole rather than the equations for the individual particles.

\begin{remark}
A question concerns the notion of the derivation can arise. Mathematically, we can postulate a kinetic equations and speak about the derivation of the particle dynamics from it. But actually, kinetic equations are known to be obtained heuristically with the use of the notion of a particle and the laws of motion for particles. 

It is more preferable to speak not about a derivation, but about a reconciliation of the kinetic dynamics with the microdynamics. Usually, the reconciliation is understood as a derivation of the kinetic equations from the microdynamics. We claim that neither kinetic dynamics can be derived from the microdynamics, nor microdynamics can be derived from the kinetic dynamics in the ``pure'' sense. Every level of nature has its own laws, but one can speak about some agreement between them.

We do not demand a kinetic equation to be derived from the microdynamics, by we say that a kinetic equation is in agreement with the microdynamics if it has microscopic solution or if it is an approximation of another kinetic equation which has microscopic solutions (see below).
\end{remark}

\section{Other kinetic equations}

\subsection{Boltzmann--Enskog--Vlasov equation}

The obtained results can be generalized to some other kinetic equations. Bogolyubov \cite{Bogol75,BogBog} noticed that the kinetic equation for hard spheres with additional  long-range pair interaction potential $\Phi(|r_i-r_j|)$ also has microscopic solutions. Such system is described by the Boltzmann--Enskog equation with an additional ``Vlasov integral'' (we will call it the Boltzmann--Enskog--Vlasov equation):
\begin{multline}\label{EqBEV}
\frac{\partial f(r_1,v_1,t)}{\partial t}=-v_1\frac{\partial f(r_1,v_1,t)}{\partial r_1}+
na^2\int_{(v_{21},\sigma)\geq0}
(v_{21},\sigma)[f(r_1,v'_1,t)f(r_1+a\sigma,v'_2,t)\\-f(r_1,v_1,t)f(r_1-a\sigma,v_1,t)]
d\sigma dv_2
+\frac nm\int\frac{\partial\Phi(|r_1-r_2|)}{\partial r_1}\rho(r_2,t)dr_2
\frac{\partial f(r_1,v_1,t)}{\partial v_1},
\end{multline}
where $m>0$ is a mass of a hard sphere, $\rho(r_2,t)=\int f(r_2,v_2,t)\,dv_2$. Let us formulate the proposition analogous to proposition~2.

\begin{proposition}
Consider the (generalized) function $f$ given by (\ref{Eqfdelta}) where the time dependences of $q_i(t,\Gamma)$ and $w_i(t,\Gamma)$ are defined by initial values $\Gamma$ and the Newtonian laws of motion: $\dot q_i(t)=w_i(t)$, and
\begin{equation}\label{EqNewtonWi}\dot w_i(t)=-\frac1m\sum_{j\neq i}\frac{\partial\Phi(|q_i(t)-q_j(t)|)}{\partial q_i(t)}\end{equation}
if $|q_i(t)-q_j(t)|>a$ for all $j\neq i$, otherwise $w_i(t)$ changes according to the law of elastic collision (\ref{EqHardColl}) (with the replacement of $v_{1,2}$ by $w_{i,j}$). This generalized function is defined as a functional on $CC^1(\Omega)$ in the same manner as in section~3 (with the correspondingly defined dynamic flow $S^{(N)}_t$). This function satisfies the Boltzmann--Enskog--Vlasov equation (\ref{EqBEV}).
\end{proposition}

A theorem analogous to theorem~1 can be formulated for this case as well. The formulation is exactly the same, except of the corresponding definition of the dynamic flow.

\begin{theorem}
1) Let us treat the Boltzmann--Enskog--Vlasov equation (\ref{EqBEV}) as an equation for functions $f\in C^1(\mathbb R^3\times\mathbb R^3\times[0,T))$, $T>0$. In this case, the equation is irreversible, i.e., if some function $f(r_1,v_1,t)$ is a solution of this equation, then, in general, the function $\widetilde f(r_1,v_1,t)=f(r_1,-v_1,-t)$ is not a solution;

2) Let us treat the Boltzmann--Enskog equation as an equation for generalized functions $f\in\Delta$ on test functions from $CC^1(\Omega)$. In this case, the equation is reversible, i.e., if some generalized function $f(r_1,v_1,t;\Gamma)$ is a solution of this equation, then the generalized function $\widetilde f(r_1,v_1,t;\Gamma)=f(r_1,-v_1,-t;\Gamma)$ is also a solution.

3) The continuously differentiable solutions $f_\varepsilon(r_1,v_1,t)$ constructed in theorem~1 (sums of delta-like functions) are reversible on time segments $[-T,T]$, where $T>0$ depends on $\varepsilon>0$.

\end{theorem}

The proof is analogous to the proof of theorem~2 (the Vlasov term is reversible in both cases, only the Boltzmann--Enskog collision integral is irreversible in the case of a continuous $f$).

Thus, we can speak about a derivation of the classical Newton equation (or Hamilton equations) from the Boltzmann--Enskog--Vlasov equation, since it has the corresponding microscopic solutions.

\subsection{Boltzmann equation for hard spheres}

Let us consider the Boltzmann equation for hard spheres. It has the same form (\ref{EqBE}) as the Boltzmann--Enskog equation, but the collision integral has the form
$$
\St f=na^2\int_{(v_{21},\sigma)\geq0}
(v_{21},\sigma)[f(r_1,v'_1,t)f(r_1,v'_2,t)-f(r_1,v_1,t)f(r_1,v_1,t)]
d\sigma dv_2.
$$

The Boltzmann equation has not microscopic solutions. It can be obtained from the Boltzmann--Enskog equation by the Boltzmann--Grad limit $a\to0$, $n\to\infty$ (or $N\to\infty$, since $n=\frac NV$ and $V=const$), $na^2=const$ (also one can prove the week convergence of the corresponding solutions \cite{BelLach88,ArkCerc}). The Boltzmann--Enskog equation is considered to be more exact than the Boltzmann one \cite{BelLach91}. The Boltzmann equation is an approximation of the Boltzmann--Enskog one in the case of a large number of hard spheres and  a negligible diameter of them. 

It is worthwhile to note that the Bogolyubov's derivation of the Boltzmann equation from the microdynamics leads not to the Boltzmann equation itself, but to some modification of it, which turns to the Boltzmann--Enskog equation in the case of a hard sphere gas \cite{Bogol}.

So, the Boltzmann equation has not microscopic solutions, but it can be represented as an approximation of an equation which has such solutions. In this sense one can speak about a reconciliation of the reversible microdynamics of hard spheres and irreversible dynamics of the Boltzmann equation for a hard sphere gas.

\subsection{Boltzmann equation in the general form}

Consider the Boltzmann equation with a more general form of the collision integral

\begin{multline*}
\frac{\partial f(r_1,v_1,t)}{\partial t}=-v_1\frac{\partial f(r_1,v_1,t)}{\partial r_1}+
na^2\int_{\mathbb R^3\times S^2}
B(v_{21},\sigma)[f(r_1,v'_1,t)f(r_1,v'_2,t)\\-f(r_1,v_1,t)f(r_1,v_1,t)]
d\sigma dv_2.
\end{multline*}
This case describes the dynamics of the particles with a pair interaction potential $\Phi$. The function (collision kernel) $B$ is determined by $\Phi$. 

This equation also has not microscopic solutions. However, as in the considered particular case of hard spheres, it would be interesting to obtain this equation as a limit of another kinetic equation which has microscopic solutions. Let us  sketch this construction.

Let $\Phi$ be an interaction potential which corresponds to the given $B$. Let $N$ particles with positions $q_i(t)$ and velocities $w_i(t)$ at the moment $t$, $i=1,\ldots,N$, interact with the pair potential $\Phi(\frac{|q_i-q_j|}a)$, $a>0$. Assume that $\Phi(r)\to0$ as $r\to\infty$.

Consider the Boltzmann--Grad limit $N\to\infty$, $a\to0$, $Na^2\to const$. Since $a\to0$ and $\Phi(r)\to0$ as $r\to\infty$, the interaction time tends to zero. Hence, the interactions can be approximately described as instantaneous. The dynamics of particles with instantaneous interactions can be described by a Boltzmann--Enskog-like equation, which has microscopic solutions.

A detailed construction of this approximation will be a subject of a further work.

\section{Conclusions}

We have formulated the existences of microscopic solutions of the kinetic Boltzmann--Enskog equation for hard spheres and the Boltzmann--Enskog--Vlasov equation for hard spheres with additional long-range interaction potential as rigorous propositions and have pointed out the importance of additional condition (\ref{EqCollCont}). We have analyzed the relation between the irreversibility of these equations and the reversibility of the microscopic solutions of them. The solutions of these kinetic equations are time-irreversible if the considered solutions belong to the class of continuous functions. They are reversible if the considered solutions have the form of sums of delta-functions (microscopic solutions have this form). So, the reversibility or irreversibility property of the dynamics depends on the considered class of functions.

Also we have constructed the continuous approximate microscopic solutions. They are reversible on bounded time intervals. The approximate microscopic solutions have the form of sums of delta-like distributions. They behave as ``spreaded hard spheres'', the limiting motion of the centers of these delta-like distributions is described by the exact microscopic solutions.

We have used these results to propose a new way of reconciliation of the kinetic equations with the microdynamics. In contrast to the traditional paradigm, where the laws of motion for individual particles are regarded as fundamental laws (reductionism), we have postulated the kinetic equations, which describes the dynamics of the system as a whole (holism). The establishment of the existence of microscopic solutions for a kinetic equation is a kind of derivation of the microdynamics. So, one can speak about a derivation of the microdynamics from the kinetic dynamics.

Finally, we have discussed the Boltzmann equation. It has not microscopic solutions, but it can be regarded as an approximation of the Boltzmann--Enskog (or a Boltzmann--Enskog-like) equation in the case of a large number of particles and a negligible interaction radius of them. In this sense one can speak about a reconciliation of the microdynamics with the Boltzmann equation as well.

\section{Acknowledgments}

The author is grateful to V.\,V.~Kozlov, S.\,V.~Kozyrev, A.\,I.~Mikhailov, A.\,N.~Pechen, H.~Spohn, V.\,V.~Vedenyapin, I.\,V.~Volovich, E.\,I.~Zelenov, and participants of the Seminar of the Department of Mathematical Physics (Steklov Mathematical Institute of the Russian Academy of Sciences) for fruitful discussions and remarks. This work was partially supported by the Russian Foundation for Basic Research (projects 11-01-00828-a and 11-01-12114-ofi-m-2011), the grant of the President of the Russian Federation (project NSh-2928.2012.1), and the Division of Mathematics of the Russian Academy of Sciences.


\begin{thebibliography}{99} 

\bibitem{Bogol75}N.\,N.~Bogolyubov, ``Microscopic solutions of the Boltzmann-Enskog equation in kinetic theory for elastic balls'', {\it Theor. Math. Phys.} {\bf 24}(2) 804--807 (1975).

\bibitem{BogBog}N.\,N.~Bogolubov, N.\,N.~Bogolubov~Jr., {\it Introduction to Quantum Statistical Mechanics} (World Scientific, Singapore, 2010).

\bibitem{Vlasov}A.\,A.~Vlasov, {\it Many-particle theory and its application to plasma} (Gordon and Breach, New York, 1961).

\bibitem{Kozlov}V.\,V.~Kozlov, {\it Thermal Equilibrium in the sense of Gibbs and Poincar\'{e}} (Institute of Computer Research, Moscow--Izhevsk, 2002) (in Russian); V.\,V.~Kozlov, ``Thermal Equilibrium in the sense of Gibbs and Poincar'{e}'' {\it Dokl. Math.} {\bf 65}(1) 125--128 (2002).


\bibitem{Bogol}N.\,N.~Bogoliubov, {\it Problems of Dynamic Theory in Statistical Physics} (Gostekhizdat, Moscow--Leningrad, 1946; North-Holland, Amsterdam, 1962; Interscience, New 
York, 1962).

\bibitem{Bogol77}N.\,N.~Bogolyubov, {\it Kinetic Equations and Green Functions in Statistical Mechanics} (Institute of Physics of the Azerbaijan SSR Academy of Sciences, Baku, 1977), Preprint N.57 (in Russian).

\bibitem{Lanford} O.\,E.~Lanford, ``Time evolution of large classical systems'', {\it Lect. Notes Phys.} {\bf 38} 1--111 (1975).

\bibitem{Lebowitz}J.\,L.~Lebowitz, ``From time-symmetric microscopic dynamics to time-asymmetric macroscopic behavior: An overview'', arXiv:~0709.0724 [cond-mat.stat-mech].

\bibitem{Spohn}H.~Spohn, ``Kinetic equations from Hamiltonian dynamics: Markovian limits'', {\it Rev. Mod. Phys.} {\bf 52}(3) 569--615 (1980).

\bibitem{Villani}C.~Villani, ``A review of mathematical topics in collisional kinetic theory'', {\it Handbook of Mathematical Fluid Dynamics} {\bf 1}71--305 (2002).

\bibitem{VolFuncMech}I.\,V.~Volovich, ``Time irreversibility problem and functional formulation of classical mechanics'', {\it Vestnik Samara State University} {\bf 8/1} (2008) 35--54; arXiv:~0907.2445v1 [cond-mat.stat-mech]

\bibitem{VolRand}I.\,V.~Volovich, ``Randomness in classical and quantum mechanics'', {\it Found. Phys.} {\bf 41}(3) 516--528 (2011); arXiv:~0910.5391v1 [quant-ph]

\bibitem{Mikh}A.\,I.~Mikhailov, ``Functional mechanics: Evolution of the moments of distribution function and the Poincar\'e recurrence theorem'', {\it $p$-Adic Numbers, Ultrametric Analysis and Applications} {\bf 3}(3) 205--211 (2011).

\bibitem{PiskVol}E.\,V.~Piskovskiy and I.\,V.~Volovich, ``On the correspondence between Newtonian and functional mechanics'', in {\it Quantum Bio-Informatics IV} (World Scientific, Singapore, 2011), 363--372.

\bibitem{Pisk}E.\,V.~Piskovskiy, ``On functional approach to classical mechanics'', {\it $p$-Adic Numbers, Ultrametric Analysis and Applications} {\bf 3}(3) 243--247 (2011).

\bibitem{TrVolFuncRat}A.\,S.~Trushechkin and I.\,V.~Volovich, ``Functional classical mechanics and rational numbers'', {\it $p$-Adic Numbers, Ultrametric Analysis and Applications} {\bf 1}(4) 361--367 (2009); arXiv:~0910.1502 [math-ph].

\bibitem{VolBogol}I.\,V.~Volovich, ``Bogoliubov equations and functional mechanics'', {\it Theor. Math. Phys.} {\bf 164}(3) 1128--1135 (2010).

\bibitem{Resib}P.~Resibois, ``$H$-theorem for the (modified) nonlinear Enskog equation'', {\it J. Stat. Phys.} {\bf 19}(6) 593--609 (1978).

\bibitem{BelLach91}N.~Bellomo and M.~Lachowicz, ``On the asymptotic theory of the Boltzmann and Enskog equations: A rigorous $H$-theorem for the Enskog equation'', {\it Springer Lecture Notes in Mathematics: Mathematical Aspects of Fluid and Plasma Dynamics} {\bf 1191} 15--30 (1991).

\bibitem{CPG}C.~Cercignani, V.\,I.~Gerasimenko, and D.\,Ya.~Petrina, {\it Many-Particle Dynamics and Kinetic Equations} (Kluwer Academic Publishing, Dordrecht, 1997).

\bibitem{Cerc}C.~Cercignani, {\it Theory and Application of the Boltzmann Equation} (Scottish Academic Press, Edinburgh, 1975; Elsevier, New York, 1975); C.~Cercignani, ``On the Boltzmann equation for rigid spheres'', {\it Transport Theory and Statistical Physics} {\bf 2}(3) 211--225 (1972).

\bibitem{Polew}J.~Polewczak, ``Global existence and asymptotic behaviour for the nonlinear Enskog equation'', {\it SIAM J. Appl. Math.} {\bf 49}(3) 952--959 (1989).

\bibitem{HaNoh}S.-Y.~Ha and S.-E.~Noh, ``New a priori estimate for the Boltzmann--Enskog equation'', {\it Nonlinearity} {\bf 19}(6) 1219--1232 (2006).

\bibitem{BelLach88}N.~Bellomo and M.~Lachowicz, ``On the asymptotic equivalence between the Enskog and the Boltzmann equations'', {\it J. Stat. Phys.} {\bf 51}(1/2) 233--247 (1988).

\bibitem{ArkCerc}L.~Arkeryd and C.~Cercignani, ``On the convergence of solutions of the Enskog equation to solutions of the Boltzmann equation'', {\it Comm. Partial Diff. Eqns.} {\bf 14}(8--9) 1071--1090 (1989); L.~Arkeryd and C.~Cercignani, ``Global existence in $L_1$ for the Enskog equation and convergence of the solutions to solutions of the Boltzmann equation'', {\it J. Stat. Phys.} {\bf 59}(3/4) 845--867 (1990).

\end{thebibliography}
\end{document}